\documentclass[12pt]{article}
\pdfoutput=1 
\usepackage[T1]{fontenc}
\usepackage[utf8]{inputenc}
\usepackage[english]{babel}
\usepackage{hyperref}
\usepackage{siunitx}
\usepackage{multicol}
\usepackage{graphicx}
\usepackage{amsmath}
\usepackage{amssymb}
\usepackage{multirow}
\usepackage{bbm}
\usepackage{bm}
\usepackage{authblk}
\usepackage{color}
\usepackage[round]{natbib}
\usepackage{fancyhdr} 
\usepackage{subfig}
\usepackage{float}
\usepackage{textgreek}
\usepackage{mathtools}
\usepackage{paralist}
\usepackage{adjustbox}
\usepackage{bbm}
\usepackage{booktabs,ctable,threeparttable}
\usepackage{enumitem}
\usepackage{amsthm}
\usepackage{pifont}
\newcommand{\diagdots}[3][-25]{%
	\rotatebox{#1}{\makebox[0pt]{\makebox[#2]{\xleaders\hbox{$\cdot$\hskip#3}\hfill\kern0pt}}}}

\theoremstyle{plain} \newtheorem{theorem}{Theorem}
\theoremstyle{plain} \newtheorem{lemma}{Lemma}
\theoremstyle{plain} \newtheorem{corollary}{Corollary}
\theoremstyle{remark}\newtheorem{remark}{Remark}
\theoremstyle{plain}\newtheorem{example}{Example}

\usepackage{paralist} 
\DeclareMathAlphabet{\mathpzc}{OT1}{pzc}{m}{it}

\allowdisplaybreaks

\DeclareMathOperator*{\argmin}{arg\,min}
\date{\vspace{-5ex}}
\begin{document}
	\def\spacingset#1{\renewcommand{\baselinestretch}%
		{#1}\small\normalsize} \spacingset{1}
	\title{\bf A unified analysis of regression adjustment in randomized experiments}
	\author{Katarzyna Reluga,\thanks{Division of Biostatistics, School of Public Health, 		University of California, Berkeley, U.S.A.  E-mail: katarzyna.reluga@berkeley.edu.}\;
    Ting Ye\thanks{Department of Biostatistics, University of Washington, 3980 15th Avenue NE, Box 351617, Seattle, WA 98195, U.S.A.  E-mail: tingye1@uw.edu.}\;
	and\;Qingyuan Zhao\thanks{Department of Pure Mathematics and Mathematical Statistics, University of Cambridge, Centre for Mathematical Sciences, Wilberforce Road, Cambridge CB3 0WB, U.K. E-mail: qyzhao@statslab.cam.ac.uk.\\
   {{The authors gratefully acknowledge support from} the Swiss National Science Foundation for the project P2GEP2-195898.}}
	}
	\maketitle

	\begin{abstract}
		 Regression adjustment is broadly applied in randomized trials under
		the premise that it usually improves the precision of a treatment
		effect estimator. However, previous work has shown that this is not
		always true. To further understand this phenomenon, we develop a
		unified comparison of the asymptotic variance of a class of linear
		regression-adjusted estimators. Our analysis is based on the
		classical theory for linear regression with heteroscedastic errors
		and thus does not assume that the postulated linear model is correct. For
		a completely randomized binary treatment, we provide 
		sufficient conditions under which some regression-adjusted
		estimators are guaranteed to be more asymptotically efficient than
		others. 
		We explore other settings such as general treatment assignment mechanisms and generalized linear models, and find that the variance dominance phenomenon no longer occurs.
	
	\end{abstract}
	
	\noindent%
	\vfill
	
	\newpage

\section{Introduction}
\label{sec:main-intro}

Randomized experiments are 
the gold standard to answer questions
about causality. Many researchers use multiple linear regression with a treatment indicator and some baseline covariates 
to analyze randomized experiments, in which the treatment coefficient is often interpreted as a causal effect. In some fields, this is known as
the ``analysis of covariance'' (ANCOVA), which was first proposed by
\citet{fisher1932statistical} to unify ``two very widely applicable
procedures known as regression and analysis of variance''.
This common practice is motivated by the belief that regression
adjustments can increase precision if covariates in the regression are
predictive of the outcome.

However, as pointed out by many authors,
this is not always true especially when there is a lot of treatment
effect heterogeneity.
Regression adjustment in randomized experiments has been studied in two different frameworks, namely the finite-population
potential outcome model
\citep{neyman1923application,rubin1974estimating} and the
super-population model that assumes the experimental units are drawn
independently from an infinite population \citep[see e.g.][Chapter
7]{imbens-rubin-2015}. 
Three estimators have been extensively studied in the literature: the simple
difference-in-means or analysis of variance (ANOVA) estimator; the
ANCOVA estimator that includes covariate main effects; and the
regression-adjusted estimator that includes covariate main effects and
all treatment-covariate interactions. The last one is termed as the
analysis of heterogeneous covariance (ANHECOVA) estimator by
\citet{ye2020principles}. The main conclusions about the asymptotic efficiency of these
estimators are the same, regardless of whether the potential outcome model
\citep{freedman2008regression,FREEDMAN2008180,schochet2010regression,lin2013agnostic,guo-basse}
or super-population model
\citep{koch1998issues,yang2001efficiency,tsiatis2008covariate,schochet2010regression,rubin2011targeted,ye2020principles}
is used. 
Consider two estimators $\hat{\beta}_1$ and
$\hat{\beta}_2$ that converge to the same limit. 
We say that $\hat{\beta}_1$ (asymptotically)
\emph{uniformly dominates} $\hat{\beta}_2$ if the (asymptotic) variance of
$\hat{\beta}_1$ is always smaller or equal than that of $\hat{\beta}_2$, no
matter what the underlying distribution is. In both the potential
outcome model and the super-population model, it has been found that
ANHECOVA uniformly dominates the other two, but, somewhat
surprisingly, ANCOVA does not uniformly dominate ANOVA. 

A major limitation of the existing 
analysis of regression adjustment is
that the investigations are restricted to specific estimators and
provides limited insights into the phenomenon of uniform dominance. 
The variance calculations are often quite
technical, which further make the theoretical results less
accessible to practitioners. Furthermore, the existing literature does not tell us whether including all treatment-covariate interactions is preferred in other cases such as stratified experiments and generalized linear models.

In this article, we provide a unified analysis for a large class of
linear-regression adjusted estimators. Besides the estimators
mentioned above, our theory also applies to regression estimators with
some coefficients fixed (such as the difference-in-differences
estimator) or with treatment-covariate interactions only. By a simple
application of the textbook theory for linear regression with
heteroscedastic errors, this analysis not only recovers the known
relationships between ANOVA, ANCOVA, and ANHECOVA, but also
immediately provides a
sufficient 
condition for uniform dominance when the
expectation of the covariates is known (see Theorem~\ref{theorem:var-order} below). In
the more practical situation when the covariate expectation is
unknown, a slightly different sufficient 
condition is obtained (see Theorem~\ref{theorem:order-non-cent} below). This condition shows that, for example, the so-called lagged-dependent-variable regression
estimator is more efficient than the difference-in-differences
estimator in randomized experiments, despite them having a bracketing
relationship in observational studies \citep{ding-li-2019}. This
unified analysis allows us to explore whether
the uniform dominance extends to more complicated settings and provide
numerical counterexamples. Some further remarks are provided at the
end of this article, whereas proofs of the technical Lemmas can be found in Appendix \ref{sec:appendix}. 

\section{Linear regression adjustment in randomized trials}
\label{sec:modelling_frame}
Consider a random sample $\{(A_i, X_i, Y_i)\}_{i=1}^n$ of $n$ units,
where $A_i \in \{0,1\}$ is a binary treatment indicator, $X_i=(X_{i1},
X_{i2}, \dots,X_{ip})^T \in \mathbb{R}^p$ is a vector of unit
covariates observed before treatment assignment, and $Y_i \in
\mathbb{R}$ is a real-valued outcome of the unit. We assume that
$(A_i,X_i^T,Y_i), i=1,\dots,n$ is independent and identically
distributed, which is often a good approximation when the units are
randomly sampled from a large population. To simplify the notation, we
drop the subscript $i$ when referring to a generic unit from the
population.

Unless mentioned otherwise, we assume that each unit receives the
treatment independently with equal probability $\mathrm{pr}(A = 1 \mid
X) = \pi$, where $0 < \pi < 1$ is a known constant. In other words,
treatment is assigned by a simple
Bernoulli trial, which approximates random sampling without
replacement that is often studied in the finite-population model
\citep{freedman2008regression,FREEDMAN2008180,lin2013agnostic}. Under
this assignment mechanism and standard assumptions in causal
inference, the average treatment effect $\beta_{\text{ATE}} = E[Y(1) - Y(0)]$, where
$Y(a)$ is the potential outcome of unit $i$ under treatment level $a$,
can be identified as
\citep[see e.g.][Chapter
7]{imbens-rubin-2015}:
\begin{equation}\label{eq:ATE}
	\beta_{\text{ATE}}=E(Y \mid A=1)-E(Y \mid A=0).
\end{equation}

In this article, we consider the following class of regression
adjusted estimators of $\beta_{\text{ATE}}$. Let $\Gamma =
\Gamma^{(1)} \times \dotsb \times \Gamma^{(p)} \subseteq \mathbb{R}^p$
and $\Delta = \Delta^{(1)} \times \dotsb
\times \Delta^{(p)} \subseteq \mathbb{R}^p$ be two user-specified
sets, where the individual components $\Gamma^{(j)}$ and
$\Delta^{(j)}$ are either the real line $\mathbb{R}$ or a
singleton. Define the constrained ordinary least squares estimator as
\begin{align}
	\hat{\theta} = (\hat{\alpha},
	\hat{\beta},\hat{\gamma},\hat{\delta})=&\argmin_{\gamma \in \Gamma,
		\delta \in \Delta}
	\frac{1}{n}\sum_{i=1}^{n}\{Y_i-\alpha-\beta
	A_i-\gamma^T
	X_i-A_i(\delta^TX_i)\}^2.\label{eq:beta_sample}
\end{align}
We sometimes use the notation $\hat{\theta}(\Gamma,\Delta)$ (and similarly for the components of $\hat{\theta}$) to emphasize the dependence of the estimator on the sets $\Gamma$ and $\Delta$. Lemma \ref{lemma:alpha_beta} in Section \ref{sec:cen_cov} shows that $\hat{\beta}$ is
a reasonable estimator of $\beta_{\text{ATE}}$ when the covariates are centered,
i.e.\ $E(X)=0$; otherwise $\beta_{\text{ATE}}$ can be estimated by $\tilde{\beta}
=\hat{\beta}+\hat{\delta}^T\bar{X}$, where $\bar{X} = \sum_{i=1}^n X_i
/ n$. Before examining the asymptotic properties of $\hat{\beta}$ and
$\tilde{\beta}$, we 
give several 
examples in the class
of estimators \eqref{eq:beta_sample}.

\begin{example}\label{example:anova}
	The ANOVA, ANCOVA, ANHECOVA estimators correspond to setting
	$\Gamma=\Delta=\{0\}$; $\Gamma = \mathbb{R}^p$ and $\Delta=\{0\}$;
	$\Gamma = \mathbb{R}^p$ and $\Delta = \mathbb{R}^p$.
\end{example}

\begin{example}\label{example:DiD}
	In some applications, the covariate vector $X$ include the baseline
	value of the response before the treatment is assigned (let us call
	it $Y_0$). For simplicity, suppose the first entry of $X$ is $Y_0$,
	so $X = (X_1 = Y_{0} , X_{2}, \dots,X_{p})^T$. The
	difference-in-differences estimator corresponds to setting
	$\Gamma = \{1\} \times \mathbb{R}^{p-1}$ and  $\Delta = \{0\}
	\times \mathbb{R}^{p-1}$, while the lagged-dependent-variable
	regression estimator corresponds to setting
	$\Gamma \subseteq \mathbb{R}^{p}$ and $\Delta = \{0\} \times
	\mathbb{R}^{p-1}$. In observational studies, these two estimators
	rely on different identification assumptions \citep{ding-li-2019}
	and may converge to different limits. In the randomized experiment
	described above, both estimators should converge to the average
	treatment effect, but we are unaware of any comparison of their
	statistical efficiency in presence of covariates besides~$Y_0$.
\end{example}
\section{A unified analysis of linear regression-adjusted estimators}\label{sec:regression_adjust}

\subsection{Covariates with known expectation}\label{sec:cen_cov}

We first consider estimation of $\beta_{\text{ATE}}$ when the covariates $X$ have known expectation. As will be seen in a moment, the proof of uniform dominance is fairly straightforward in this case.

Consider the population counterpart to \eqref{eq:beta_sample}:
\begin{align}
	\theta =
	(\alpha,\beta,\gamma,\delta)=&\argmin_{\gamma \in \Gamma, \delta \in \Delta}
	E\{Y-\alpha-\beta A-\gamma^T
	X-A(\delta^TX)\}^2.\label{eq:beta_pop}
\end{align}
Clearly, $\theta = \theta(\Gamma, \Delta)$, and 
we often suppress the dependence of $\theta$ on $(\Gamma, \Delta)$ if it is
clear from the context.

\begin{lemma}\label{lemma:alpha_beta}
	For any $\Gamma$ and $\Delta$ of the form described in Section
	\ref{sec:modelling_frame}, we have $\beta = \beta_{\text{ATE}} - \delta^T
	E(X)$.
\end{lemma}

Without loss of generality, we shall assume that $E(X) = 0$ for the rest of Theorem \ref{sec:cen_cov}; otherwise, we can simply replace $X$ with $X -
E(X)$ since $E(X)$ is known. When $E(X) = 0$, Lemma
\ref{lemma:alpha_beta} shows that $\hat{\beta}$ is a
reasonable estimator of $\beta_{\text{ATE}}$. To study
the asymptotic properties of $\hat{\beta}$, we first state a
classical result for linear regression with heteroskedastic error. For
a proof of this result, see e.g.\ \citet{white1980heteroskedasticity}.

\begin{lemma} \label{lemma:cons-an}
	Consider a linear regression of an independent and identically distributed sample of response $Y \in \mathbb{R}$ on regressors $Z \in
	\mathbb{R}^p$. Let $\hat{\theta}$ and $\theta$
	be sample and population least squares estimators and $\epsilon(\theta) = Y - \theta^T Z$. Suppose that $E(ZZ^T)$ and $E\{\epsilon(\theta)^2ZZ^T\}$ are positive definite and $Y$, $Z$ have bounded fourth moments. Then, as $n\rightarrow\infty$, $\hat{\theta}\xrightarrow[]{}\theta$
	in probability and
	\begin{equation}\label{eq:asymp_normality}
		\sqrt{n}\left(\hat{\theta}-\theta\right)\xrightarrow[]{}N
		\left(0, \{E(ZZ^T)\}^{-1}E(ZZ^T\epsilon(\theta)^2)\{E(ZZ^T)\}^{-1}
		\right) ~ \text{in distribution.}
	\end{equation}
\end{lemma}

Note that these results do not require that the linear model is correctly specified. By applying Lemma \ref{lemma:cons-an} to our problem with an appropriate $Z$ and regression error
\begin{equation}\label{eq:error}
	\epsilon = \epsilon(\theta) = Y-\alpha-\beta
	A-\gamma^T X-A(\delta^TX),
\end{equation}
we obtain the 
expression for the asymptotic variance of
$\hat{\beta}$. The proof of this result is straightforward
due to the block diagonal structure of $E(Z Z^T)$. This is made
possible by the 
assumption that $E(X) = 0$.

\begin{lemma} \label{lem:beta-var}
	Suppose that $E(X)=0$ and the regularity conditions in Lemma
	\ref{lemma:cons-an} are satisfied. Then, as $n\rightarrow \infty$,
	we have
	\begin{equation*}
		\sqrt{n}(\hat{\beta}-\beta)\xrightarrow[]{}N \left( 0,
		\frac{E\{(A-\pi)^2\epsilon^2 \}}{\pi^2(1-\pi)^2} \right)~
		\text{in distribution.}
	\end{equation*}
\end{lemma}

To state our first main result about uniform
dominance, we introduce an additional notation. Let
$\mathcal{U}(\Gamma) \subseteq
\{1,\dotsc,p\}$ denote the unrestricted dimensions of
$\Gamma$, i.e. $\mathcal{U}(\Gamma) = \{1 \leqslant j \leqslant p: \Gamma^{(j)} =
\mathbb{R}\}$. Similarly, let $\mathcal{U}(\Delta)$ denote the unrestricted
dimensions of $\Delta$.

\begin{theorem}\label{theorem:var-order}
	Suppose $E(X)=0$. Consider two estimators $\hat{\beta}_1$ and
	$\hat{\beta}_2$ obtained
	from the least squares problem \eqref{eq:beta_sample} with
	$(\Gamma,\Delta) = (\Gamma_1,\Delta_1)$ and
	$(\Gamma_2,\Delta_2)$, respectively, and suppose
	$(\Gamma_1,\Delta_1) \neq (\Gamma_2,\Delta_2)$. Then $\hat{\beta}_1$
	uniformly dominates $\hat{\beta}_2$ if 
	\begin{equation}
		\label{eq:nested}
		\Gamma_1 \supseteq \Gamma_2,~\Delta_1 \supseteq
		\Delta_2,~\text{and either}~\pi = 1/2~\text{or}~
		\mathcal{U}(\Delta_1) \supseteq
		\mathcal{U}(\Gamma_1).
	\end{equation}
\end{theorem}
\begin{proof}
	The first-order condition for
	the least squares problem
	\eqref{eq:beta_pop} can be written as
	\begin{equation}
		\label{eq:eps-uncorrelated}
		E(\epsilon) = 0,~ E (\epsilon A) = 0,~ E\{\epsilon X_{\mathcal{U}(\Gamma)}\} = 0,~    E\{\epsilon A X_{\mathcal{U}(\Delta)}\}
		= 0.
	\end{equation}
	Let $\theta_k = (\alpha_k,\beta_k,\gamma_k,\delta_k)$ be the
	solution for when $(\Gamma,\Delta) = (\Gamma_k,\Delta_k)$ and
	$\epsilon_k = Y - \alpha_k - \beta_k A - \gamma_k^T X - A
	(\delta_k^T X)$ be the corresponding regression error, $k=1,2$. By
	Lemma \ref{lemma:alpha_beta}, $	\epsilon_2 = \epsilon_1 + (\gamma_1 - \gamma_2)^T X + A (\delta_1 -
	\delta_2)^TX$.  
	Let $V_k = [E\{(A-\pi)^2\epsilon_k^2 \}]/\{\pi^2(1-\pi)^2\}$. Then by Lemma \ref{lem:beta-var},
	\begin{align}
		\pi^2(1-\pi)^2(V_2 - V_1)
		=& E\{(A-\pi)^2(\epsilon_2^2-\epsilon_1^2)\}\nonumber\\
		=&E\left[(A-\pi)^2
		2\epsilon_1  \{(\gamma_1 - \gamma_2)^T X + A (\delta_1 -
		\delta_2)^TX \}\right] \nonumber\\
		&+ E\left[(A-\pi)^2
		\{(\gamma_1 - \gamma_2)^T X + A (\delta_1 -
		\delta_2)^TX\}^2\right] \nonumber \\
		\geq& 2 E\left[(A-\pi)^2
		\epsilon_1  \{(\gamma_1 - \gamma_2)^T X + A (\delta_1 -
		\delta_2)^TX \}\right]. \label{eq:comparison_1}
	\end{align}
	Since $\mathcal{U}(\Gamma_k)$ contains the unrestricted dimensions of $\gamma_k$, $k=1,2$, and $\Gamma_1 \supseteq \Gamma_2$ by assumption, the non-zero elements of $\gamma_1 - \gamma_2$ can only appear in $\mathcal{U}(\Gamma_1)$ (otherwise the coefficients are fixed by design). Similarly, the non-zero elements of $\delta_1 - \delta_2$ 
	appear in $\mathcal{U}(\Delta_1)$. By using 
	\eqref{eq:eps-uncorrelated} and
	$A \perp \!\!\! \perp X$, $A^2 = A$, we have
	\[
	\begin{split}
		&E\left[(A-\pi)^2
		\epsilon_1  \{(\gamma_1 - \gamma_2)^T X + A (\delta_1 -
		\delta_2)^TX \}\right] \\
		=& (\gamma_1 - \gamma_2)_{\mathcal{U}(\Gamma_1)}^T \left\{
		E\left[ (1 - 2 \pi) \epsilon_1 A
		X_{\mathcal{U}(\Gamma_1)} \right] + E\left[ \pi^2 \epsilon_1
		X_{\mathcal{U}(\Gamma_1)} \right] \right\} \\
		&+ (\delta_1 - \delta_2)_{\mathcal{U}(\Delta_1)}^T E\left[ (1 -
		\pi)^2 \epsilon_1 A X_{\mathcal{U}(\Delta_1)} \right] 
		= (\gamma_1 - \gamma_2)_{\mathcal{U}(\Gamma_1)}^T
		E\left[ (1 - 2 \pi) \epsilon_1 A
		X_{\mathcal{U}(\Gamma_1)} \right],
	\end{split}
	\]
	where the last equality follow from applying
	\eqref{eq:eps-uncorrelated} to $\epsilon = \epsilon_1$ and $(\Gamma,
	\Delta) = (\Gamma_1,\Delta_1)$. Finally, $E\left[ (1 - 2 \pi) \epsilon_1 A
	X_{\mathcal{U}(\Gamma_1)} \right] = 0$ if $\pi = 1/2$ or
	$\mathcal{U}(\Delta_1) \supseteq \mathcal{U}(\Gamma_1)$ by
	\eqref{eq:eps-uncorrelated}.
\end{proof}
In words, Theorem \ref{theorem:var-order} says that, when the expectation of the covariates is known, one linear regression-adjusted estimator is uniformly dominated by another if the two linear models are nested, the first estimator is obtained from the larger model, and 
the larger model includes an interaction term whenever the corresponding main effect is present; there is no such requirement for the smaller model. The conditions in~\eqref{eq:nested} can 
be easily applied to obtain variance orderings among the examples in
Section \ref{sec:modelling_frame}. We will discuss them in more detail
after deriving a similar sufficient 
condition 
when the expectation of $X$ is unknown.

\subsection{Covariates with unknown expectation}\label{sec:cov_non_centre}

In most practical situations, we do not know the expectation of the
covariates and it is common to centre the covariates empirically
before performing the linear regression. Let $\tilde{\theta}$ be the
least squares estimator in~\eqref{eq:beta_sample} with $X_i$
replaced by $X_i-\bar{X}$ where $\bar{X}=\sum_{i=1}^{n} X_i/n$, that is,
\begin{align}
	\tilde{\theta} = (\tilde{\alpha},
	\tilde{\beta},\tilde{\gamma},\tilde{\delta})=&\argmin_{\gamma \in
		\Gamma,\delta \in \Delta}\frac{1}{n}\sum_{i=1}^{n}\left[Y_i-\alpha-\beta A_i-\gamma^T (X_i-\bar{X})-A_i\left\{\delta^T(X_i-\bar{X})\right\}\right]^2.\label{eq:beta_sample_nc}
\end{align}
We have
$\tilde{\beta}=\hat{\beta}$ if no
interaction term is included, i.e. if $\Delta = \{0\}^p$, because both \eqref{eq:beta_sample} and \eqref{eq:beta_sample_nc}
include an intercept term. More generally, by differentiating
\eqref{eq:beta_sample_nc} with respect to $\alpha$ and $\beta$ and
following the proof of Lemma \ref{lemma:alpha_beta}, it is
straightforward to verify that $\tilde{\beta}
=\hat{\beta}+\hat{\delta}^T\bar{X}$. Thus, $\tilde{\beta}$ is a
reasonable estimator of $\beta_{\text{ATE}} = \beta + \delta^T E(X)$. Estimator $\tilde{\theta}$ is invariant to any shift transformation of
the covariates. In other words, $\tilde{\theta}$ remains the same if we
replace $X_i$ by $X_i + c$, $i=1,\dotsc,n$, for any $c \in
\mathbb{R}^p$. Therefore, the statistical properties of
$\tilde{\theta}$ do not depend on $E(X)$ and, for simplifying the analysis, we shall assume $E(X) =
0$ without loss of generality. The asymptotic variance of $\tilde{\beta}$ generally differs
from that of $\hat{\beta}$ due to the variability in $\bar{X}$. The
next result quantifies this difference and it is proved in Appendix~\ref{sec:proof_lemma4}.

\begin{lemma}\label{lemma:var-beta-jk-tilde}
	Consider two estimators $\hat{\beta} = \hat{\beta}(\Gamma,\Delta)$
	and $\tilde{\beta} = \tilde{\beta}(\Gamma,\Delta)$ for some
	user-specified $(\Gamma,\Delta)$. Under regularity conditions and as
	$n\rightarrow\infty$, we have
	\[
	n\left\{\text{var}(\tilde{\beta})-\text{var}(\hat{\beta})\right\}\xrightarrow[]{}
	\delta_s^T\Sigma(2 \delta_f - \delta_s),
	\]
	where $\delta_f = \delta_0(\mathbb{R}^p,\mathbb{R}^p)$ and $\delta_s =
	\delta_0(\Gamma,\Delta)$ are obtained by solving the population
	least squares problem \eqref{eq:beta_pop} for the full model and a
	sub-model, respectively.
\end{lemma}

Let $\hat{\theta}_k$ and $\tilde{\theta}_k$ be the solution to
\eqref{eq:beta_sample} and \eqref{eq:beta_sample_nc}, respectively,
for the choice $(\Gamma,\Delta) = (\Gamma_k,\Delta_k)$, $k = 1,2$. Let
$\tilde{V}_k$ be the asymptotic variance of $\tilde{\beta}_k$, that
is, $\tilde{V}_k = \lim_{n \to \infty} n \text{var}(\tilde{\beta}_k)$. Our
main goal is to derive conditions on $(\Gamma_1,\Delta_1)$ and
$(\Gamma_2,\Delta_2)$ such that $\tilde{V}_1$ and $\tilde{V}_2$ admit
a deterministic ordering. It seems natural to require that the models
are nested: $\Gamma_1 \supseteq \Gamma_2$ and $\Delta_1 \supseteq
\Delta_2$. 

\begin{table}[htb]
	\centering
     \def~{\hphantom{0}}{
		\begin{tabular}[h]{ll c c }\hline
			Model 1 & Model 2 & $V_1 \leq V_2$ & $\tilde{V}_1 \leq\tilde{V}_2$ \\\hline
			$\sim 1 + A + X + A:X$ & $\sim 1 + A$ & True: Theorem \ref{theorem:var-order} &                                                                     True: Theorem \ref{theorem:order-non-cent}\\
			$\sim 1 + A + X + A:X$ & $\sim 1 + A + X$ & True: Theorem \ref{theorem:var-order} &                                                                                   True: Theorem \ref{theorem:order-non-cent}  \\
			$\sim 1 + A + X + A:X$ & $\sim 1 + A + A:X$ & True: Theorem \ref{theorem:var-order} & True: Theorem \ref{theorem:order-non-cent}  \\
			$\sim 1 + A + X$ & $\sim 1 + A$ &
			\multicolumn{2}{c}{Not always true: eq. \eqref{eq:anova-ancova}} \\
			$\sim 1 + A + A:X$ & $\sim 1 + A$ & {True: eq. \eqref{eq:comparison_2}} &
			{Not always true: eq. \eqref{eq:comparion3}}\\\hline
	\end{tabular}}
	\caption{Variance ordering of estimators of $\beta_{\text{ATE}}$ in \eqref{eq:ATE}. $V_k$: variance of $\hat{\beta}_k$, $k=1,2$; $\tilde{V}_k$: variance of $\tilde{\beta}_k$, $k=1,2$. }
	\label{tab:table-comparisons}
\end{table}

Table \ref{tab:table-comparisons} provides a list of uniform dominance
relationships in some basic comparisons. We use the \texttt{R}
convention to denote linear models: explanatory variables in the
regression are joined by $+$, $1$ stands for the intercept term, and
$A:X$ stands for the treatment-covariate interaction. Using Table \ref{tab:table-comparisons}, we conjecture that in order for $\tilde{V}_1 \leqslant
\tilde{V}_2$, the third condition in \eqref{eq:nested} needs to be modified. 
which is verified in the next
theorem.

\begin{theorem}\label{theorem:order-non-cent}
	Consider two estimators $\tilde{\beta}_1$ and $\tilde{\beta}_2$
	obtained from solving \eqref{eq:beta_sample_nc} with
	$(\Gamma,\Delta) = (\Gamma_1,\Delta_1)$ and
	$(\Gamma_2,\Delta_2)$, respectively, and suppose
	$(\Gamma_1,\Delta_1) \neq (\Gamma_2,\Delta_2)$. Then $\tilde{\beta}_1$
	uniformly dominates $\tilde{\beta}_2$ 
	if 
	\begin{equation}  \label{eq:nested-2}
		~\Gamma_1 \supseteq \Gamma_2,~\Delta_1 \supseteq
		\Delta_2,~\text{and}~\mathcal{U}(\Gamma_1)=\mathcal{U}(\Delta_1);
	\end{equation}
\end{theorem}
\begin{proof}
	Let $d_{\gamma} = \gamma_1 - \gamma_2$ and $d_{\delta} = \delta_1 -
	\delta_2$. By applying \eqref{eq:comparison_1} and Theorem \ref{lemma:var-beta-jk-tilde} 
	with the first model treated as the
	full model,
	\begin{align}
		\tilde{V}_2 - \tilde{V}_1 &= (V_2 - V_1) + (\tilde{V}_2 - V_2) -
		(\tilde{V}_1 - V_1) \nonumber\\
		&= \frac{1}{\pi^2 (1-\pi)^2} E\left[(A-\pi)^2
		\{d_{\gamma}^T X + A d_{\delta}^TX\}^2\right] + \delta_2^T \Sigma (2 \delta_1 -
		\delta_2) - \delta_1^T \Sigma \delta_1\nonumber\\
		&= \frac{E\left[A(A-\pi)^2
			\{d_{\gamma}^T X + d_{\delta}^TX\}^2\right] +
			E\left[(1-A)(A-\pi)^2
			\{d_{\gamma}^T X\}^2\right]}{\pi^2 (1-\pi)^2} -
		d_{\delta}^T \Sigma d_{\delta} \nonumber\\
		&= \frac{1}{\pi} (d_{\gamma} + d_{\delta})^T \Sigma
		(d_{\gamma} + d_{\delta}) + \frac{1}{1 - \pi}
		d_{\gamma}^T \Sigma d_{\gamma} - d_{\delta}^T
		\Sigma d_{\delta} \nonumber\\
		&= \frac{1}{\pi(1-\pi)} \left\{ d_{\gamma}^T \Sigma
		d_{\gamma} + (1-\pi)^2 d_{\delta}^T \Sigma
		d_{\delta} + 2 (1 - \pi) d_{\gamma}^T \Sigma
		d_{\delta} \right\} \nonumber\\
		&= \frac{1}{\pi(1-\pi)} \{d_{\gamma} + (1-\pi)
		d_{\delta}\}^T \Sigma \{d_{\gamma} + (1-\pi)
		d_{\delta}\}\geqslant 0. \label{eq:diff_theorem2}
	\end{align}
	which completes the proof.
\end{proof}
By verifying the conditions in
\eqref{eq:nested-2}, we 
have the following results
concerning the examples in Section \ref{sec:modelling_frame}.
\begin{corollary}
	The ANHECOVA estimator is asymptotically more efficient than the
	ANOVA and ANCOVA estimators. There is no guaranteed variance
	ordering between ANOVA and ANCOVA.
\end{corollary}
\begin{corollary}
	The lagged-dependent-variable regression estimator is more efficient
	than the difference-in-differences estimator.
\end{corollary}
\begin{remark}
	The condition in \eqref{eq:nested-2} might be further weakened when $\pi = 1/2$. In particular, the difference in \eqref{eq:diff_theorem2}
	is exactly $0$ if in addition, $\Gamma_1 =
	\Gamma_2$ and $\pi = 1/2$. To show this, by differentiating~\eqref{eq:beta_pop} with respect to
	$\gamma_{\mathcal{U}(\Gamma_1)}$, we have $	E\left[ X_{\mathcal{U}(\Gamma_1)} \{Y - \alpha - \beta A -
	\gamma_k^T X - A (\delta_k^T X)\} \right] = 0,~k=1,2$.
	By subtracting the two equations, we obtain
	\begin{equation}\label{eq:eq_T2_proof}
		E\left[ X_{\mathcal{U}(\Gamma_1)} (X^T d_{\gamma} + A X^T
		d_{\delta}) \right] = 0.
	\end{equation}
	Because $\Gamma_1 = \Gamma_2$ and $\Delta_1 \supseteq \Delta_2$, we
	have $\gamma_{1j} = \gamma_{2j}$ and $\delta_{1j} = \delta_{2j}$ for
	$j \not \in \mathcal{U}(\Gamma_1) = \mathcal{U}(\Delta_1)$. Thus
	$d_{\gamma,j} = d_{\delta,j} = 0$ when $j \not \in
	\mathcal{U}(\Gamma_1)$. Together with 
	equation \eqref{eq:eq_T2_proof},
	this shows that $\Sigma (d_{\gamma} + \pi d_{\delta}) =
	0$. Therefore, if we have $\pi = 1/2$ in addition, the difference
	$\tilde{V}_2 - \tilde{V}_1 = 0$. 
	In other words, when $\pi = 1/2$,
	adding or removing (more precisely, unrestricting or restricting)
	an interaction term $A X_j$ when corresponding the main effect $X_j$
	is already present in the model does not change the asymptotic
	variance of $\tilde{\beta}$.
\end{remark}
\section{Variance ordering beyond linear regression}
One can establish 
the uniform dominance 
in case of more sophisticated randomization schemes. In particular, our derivations extend to stratified randomization experiments with units grouped into $B$ strata which is in alignment with the results of \cite{liu2020regression}. The authors considered two asymptotic regimes with the number of strata $B$ or the their sample sizes to increase with the total number of units $n$. Let $S$ 
be the indicator variable of whether unit $i$ belongs to stratum $j$. In this setting, for the uniform dominance to hold, the treatment assignment should be completely randomized within strata with allocation probability equal across strata, and the estimator should be obtained using weighted regression including centred variables $A, S, X$, as well as interactions of $A$ with $S$ and $A$ with $X$.  

In contrast, the results do not usually extend to a more general assignment mechanism which depends on $X$ with $\pi(x) = p(A=1 \mid X=x)$ (see Table \ref{tab:emp_results}  below for counterexamples). Within this framework, we can identify $\beta_{\text{ATE}}$ 
for all $\Gamma, \Delta$, 
	but 
	$\hat{\beta}$ in \eqref{eq:beta_sample}, \eqref{eq:beta_pop} and \eqref{eq:beta_sample_nc} do not converge to $\beta_{\text{ATE}}$. 
	Assuming this scenario one can use weighted estimators \citep{stuart2011use,tao2019doubly} to recover \eqref{eq:ATE}. 

The result in Lemma \ref{lemma:alpha_beta} 
is tightly related to the properties of orthogonal projections and does not carry over to a wider class of generalized linear model, even if the link function we use is collapsible \citep[a link function is collapsible if including independent regressors does not change the population regression coefficients, see][for more detials]{greenland1999,collapsible}. The most common collapsible link functions in 
this setting are identity and $\log$ function; the latter is the canonical link for Poisson regression and is frequently applied. Yet, if the link is collapsible but not the identity, including the treatment-covariate interaction may identify a different estimand. 
Thus, an attempt to seek the uniform dominance 
by comparing the asymptotic variance of $\hat{\beta}$ 
seems to be an ill-posed research question in this setting.
\section{Simulation study}\label{sec:simulation_study}
We carry out numerical simulation study to verify our theoretical developments in previous sections and 
explore scenarios under which the uniform dominance does not hold (see discussion in Section \ref{sec:disc}). 
We consider scenarios with potential outcomes generated from normal and Poisson distribution. In all scenarios, the covariate is centred and normally distributed $X = X_{nc}-\bar{X}$, $X_{nc}\sim N(2,1)$ whereas $\bar{X}$ is a sample mean of observations; the observed outcome is $Y = A Y(1) +  (1-A) Y(0)  $; the treatment is assigned using a Bernoulli trial $A\sim Bernoulli(\pi)$. For linear regression, errors are normally distributed  $e_1, e_0\sim N(0,1)$. 
We consider the sample size of $n=1000$.  Below we describe simulation scenarios. 
\vspace{-0.1cm}
\begin{description}
	\item[Scenario 1] Treated outcomes: $Y(1)\sim 5 +2.5  AX + e_1$, untreated outcomes: $Y(0)\sim 3 + X + e_0$.
	\item[Scenario 2] Treated outcomes: $Y(1)\sim Poisson (\mu_1)$, untreated outcomes: $Y(0)\sim  Poisson (\mu_0)$, $\mu_1 =  \exp(3 + 0.6  X  A)$, $\mu_0 = \exp(1 + 0.6  X)$.
	\item[Scenario 3] Treated outcomes: $Y(1)\sim 7 + X + e_1$,  untreated outcomes: $Y(0)\sim 2 - X + X^2 + e_0$.
	\item[Scenario 4] The same as scenario 3, but with weights: $w = (\pi(X)(1-\pi(X))^{-1}$.
\end{description}
\vspace{-0.1cm}
To study the performance of the estimators of $\beta_{\text{ATE}}$, we calculated an average bias and a standard deviation over $1000$ Monte Carlo replications. Finally, under the Poisson $\log$ model, true values of $\beta_{\text{ATE}}$ was approximated by the difference of the large sample average ($n = 10^7$) of potential outcomes.

\begin{figure}[htb]
	\centering
	\subfloat{\includegraphics[width=0.45\textwidth]{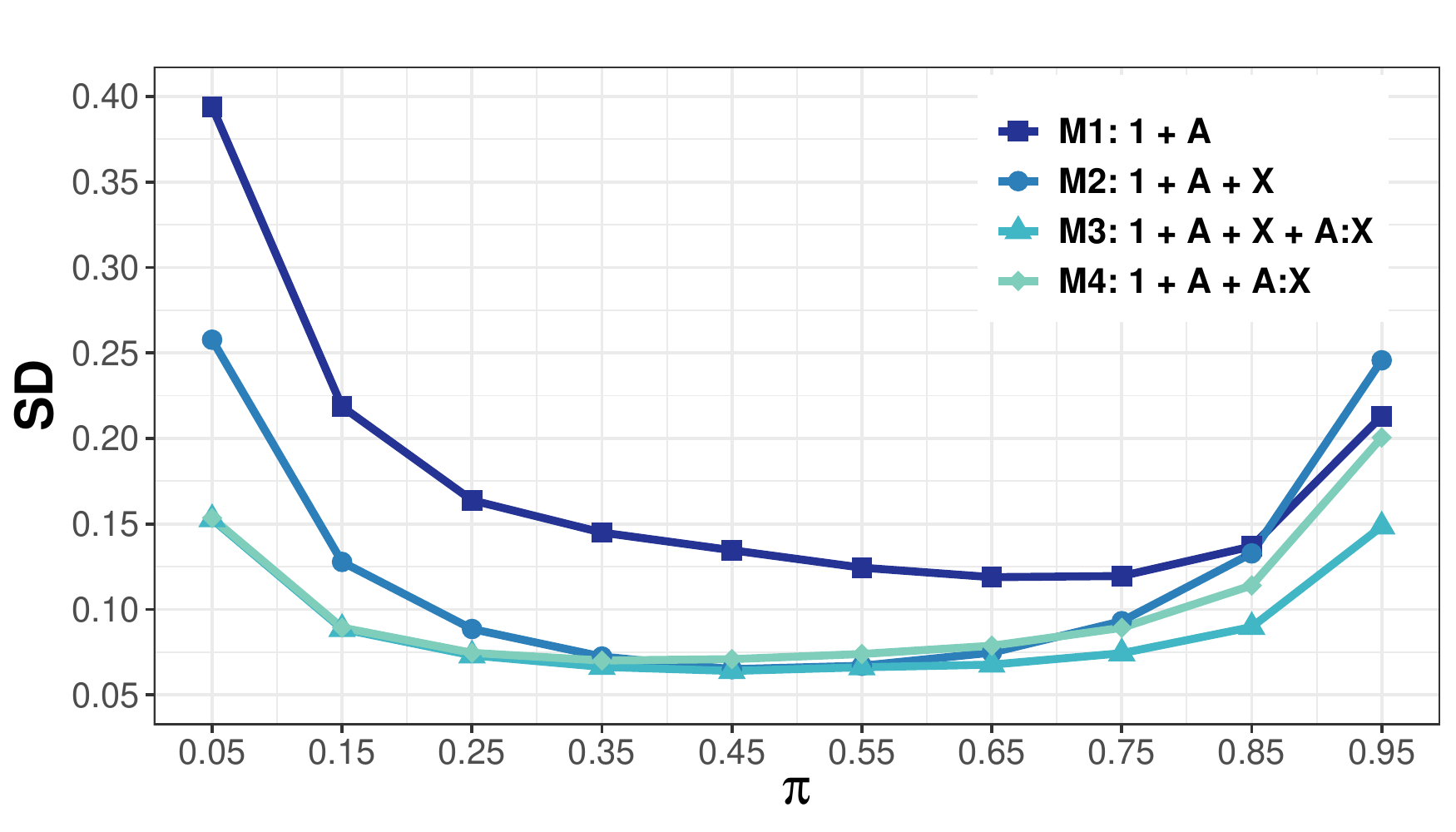}}
	\subfloat{\includegraphics[width=0.45\textwidth]{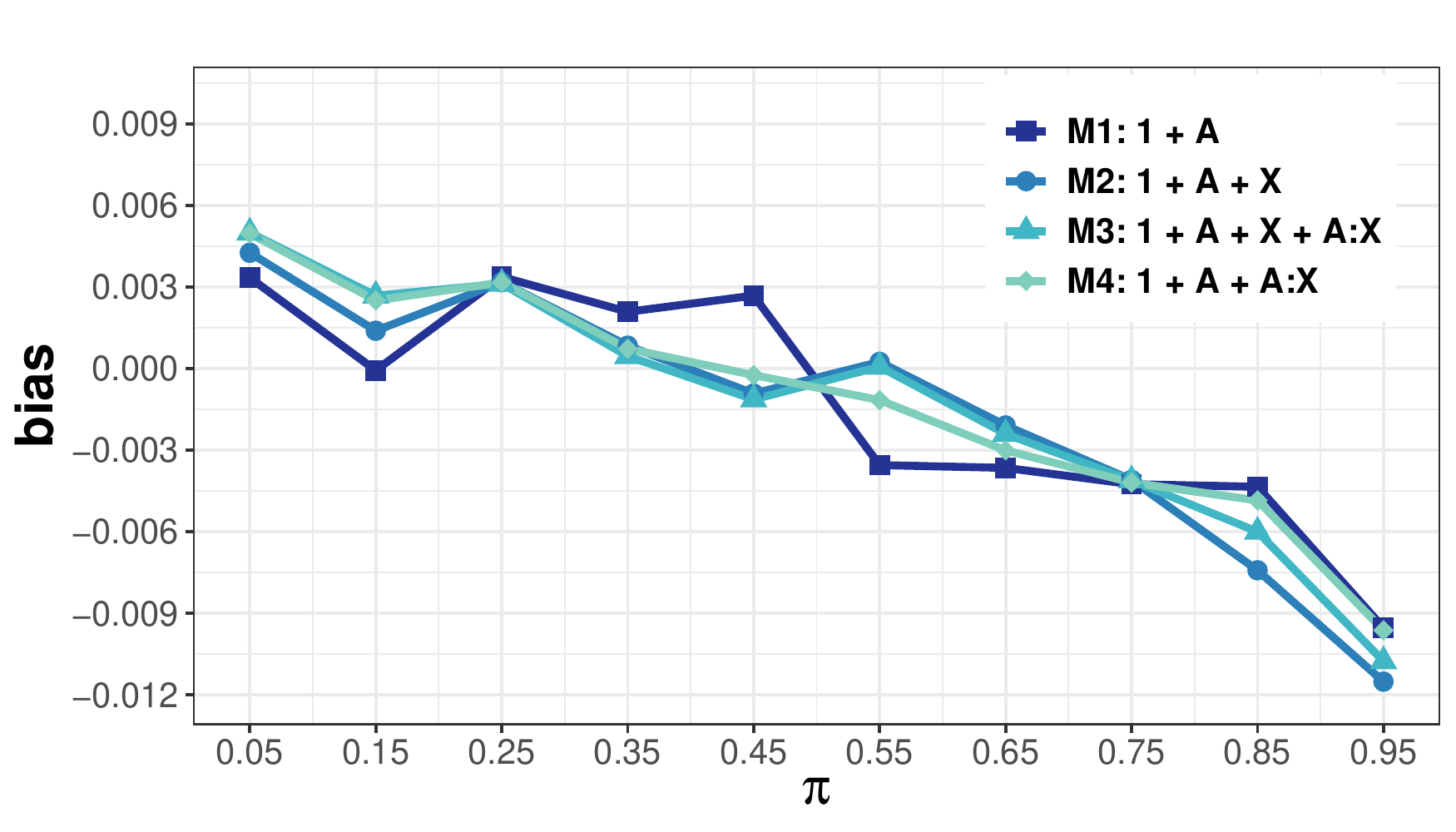}}\\
	\subfloat{\includegraphics[width=0.45\textwidth]{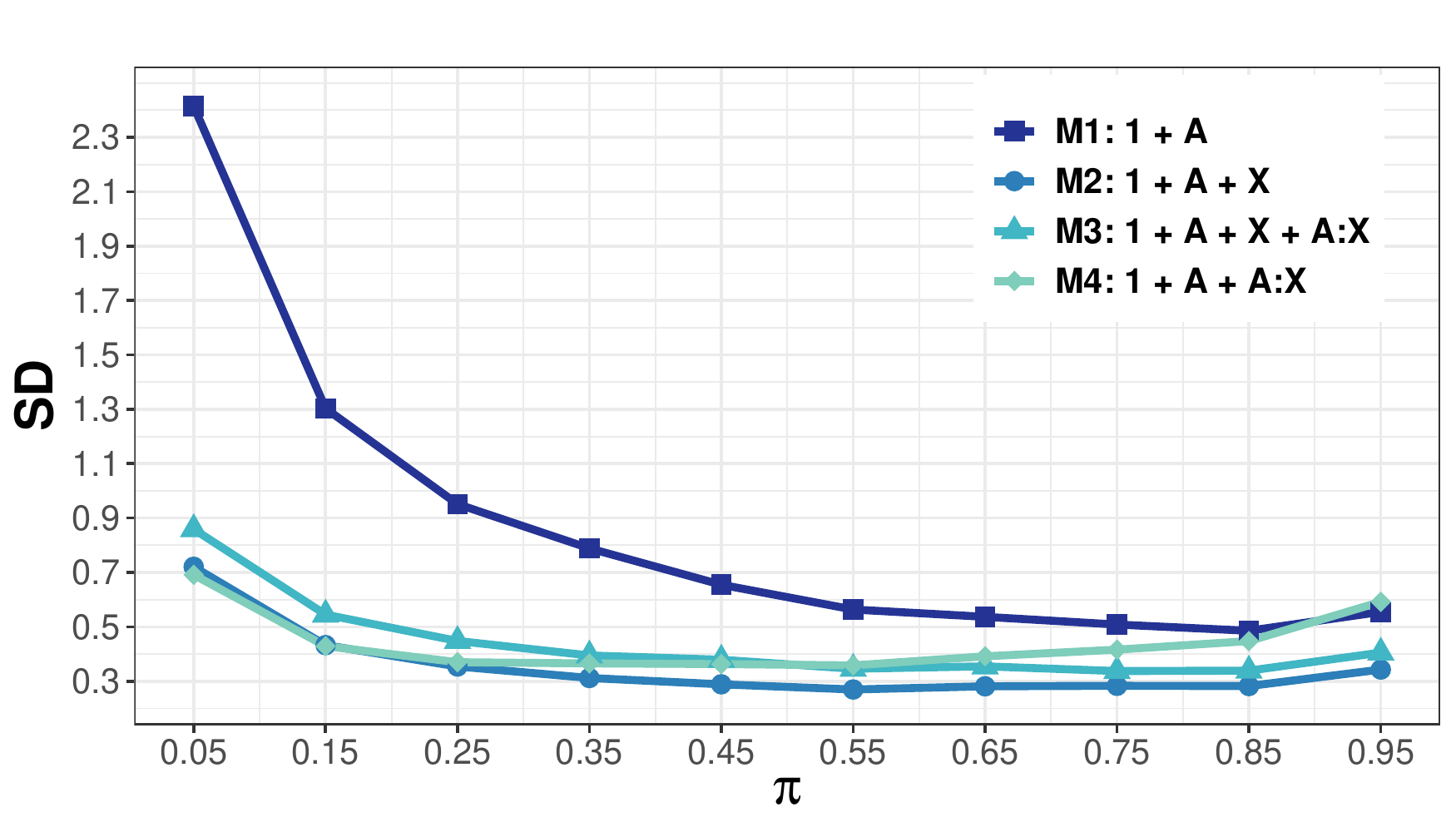}}   
	\subfloat{\includegraphics[width=0.45\textwidth]{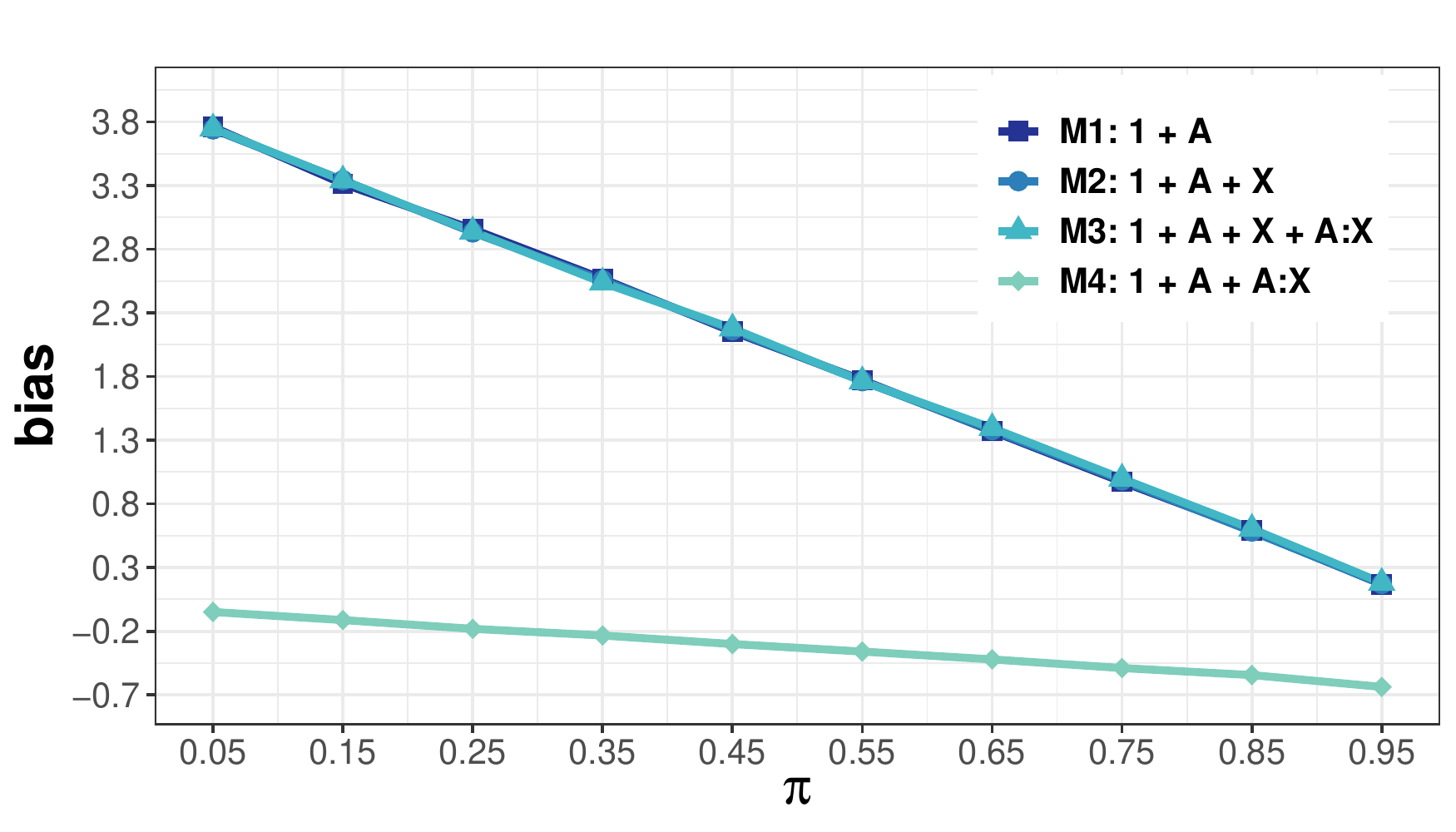}} 
	\captionof{figure}{Performance of estimators of $\beta_{\text{ATE}}$ over different values of $\pi$: 
		Scenario 1 (top panels), 
		Scenario 2 (bottom panels), SD: standard deviation.} 
	\label{fig1}
\end{figure}

Figure~\ref{fig1} shows results of simulations under Scenario 1 (top panels) and Scenario 2 (bottom panels). As for normally distributed outcomes, the estimates of $\beta_{\text{ATE}}$ are almost unbiased assuming any model (top-right panel). On the other hand, the standard deviation of $\hat{\beta}$ is the smallest for the model with both $X$ and $AX$, nevertheless adjusting for covariate only it is not beneficial for higher values of $\pi$. When it comes to the results under the Poisson $\log$ model, $\hat{\beta}$ is not the estimator of $\beta_{\text{ATE}}$ and the uniform dominance results do not hold.  
\begin{table}[htb]
		\centering
	\def~{\hphantom{0}}{
		\begin{tabular}{ccc ccc ccc}\hline
			\multicolumn{4}{c}{Scenario 3} & & \multicolumn{4}{c}{Scenario 4} \\
			\multicolumn{4}{c}{ $\beta_{\text{ATE}}  = 4$, $logit\{\pi(X)\}  = 4-2X$} & & \multicolumn{4}{c}{$\beta_{\text{ATE}}  = 4$, $logit\{\pi(X)\}  = 4-2X$, weights}  \\\hline
			\multicolumn{2}{l}{Model} & bias   &  SD &&     \multicolumn{2}{l}{Model} & bias   &  SD \\
			\multicolumn{2}{l}{$\sim 1 + A + X + A:X$} &  603  &  143 &&    \multicolumn{2}{l}{$\sim 1 + A + X + A:X$} & 153   & 440 \\
			\multicolumn{2}{l}{$\sim 1 + A  $} &  595  & 128  &&    \multicolumn{2}{l}{$\sim 1 + A $} &  -7  & 508  \\\hline
		\end{tabular}}
	\caption{Performance of $\hat{\beta}$. 	SD, standard deviation; all numerical entries are multiplied by 1000 and rounded to the nearest integer.}
	\label{tab:emp_results}
\end{table}
Table \ref{tab:emp_results} displays numerical performance of $\hat{\beta}$ 
under Scenarios 3 and 4 for which our theory does not hold. Unsurprisingly, assuming Scenario 3 with an assignment mechanism which depends on covariate $X$, $\hat\beta$ suffers from a substantial bias and the uniform dominance does not apply. Assuming Scenario 4, which involves weighting, the regression adjustment decreases the standard deviation of $\hat\beta$, but the estimator in the larger model suffers from a considerable bias. 

\section{Discussion} \label{sec:disc}
Linear regression model is still widely used 
to estimate the average treatment effect in hope to increase the precision of the estimator. 
We re-established and generalized previous results on linear-regression adjusted estimators under possible model misspecification by providing a simplified and more accessible proof of uniform dominance.   
Yet, our proof has a geometric element that exploits the linearity of the regression adjustment, and this cannot be extended to other settings. 
Thus, the phenomenon of the efficiency gain seems to be limited to the estimation problems which fit into  the linear framework and to the treatment assignment mechanisms which do not depend on $X$. 

\appendix
	\section{Proofs}\label{sec:appendix}
	\subsection{Proof of Lemma \ref{lemma:alpha_beta}}
	\begin{proof}
		Observe that $\alpha$ and $\beta$ are always unrestricted in
		\eqref{eq:beta_pop}. By taking partial derivatives with respect to
		$\alpha$ and $\beta$, we obtain
		\[
		E\{Y-\alpha-\beta A - \gamma^T X - A (\delta^T X) \}=0
		~\text{and}~ E\{A(Y-\alpha-\beta A - \gamma^T X - A (\delta^T X)
		)\}=0.
		\]
		By multiplying the first equation by $\pi$ and subtracting the
		second equation, we obtain
		\begin{align*}
			&\pi E(Y) - E(AY) + \{E(A^2) - \pi E(A)\}\beta \\
			&+ \{E(AX) - \pi E(X)\}^T \gamma + \{E(A^2 X) - \pi E(AX)\}^T\delta = 0.
		\end{align*}
		Finally, by using the assumption that the treatment
		is randomized i.e.\ $A\perp \!\!\! \perp X$ and $E(A) = E(A^2) =
		\pi$, we find that
		\begin{align*}
			\beta &= - \frac{\{\pi E(Y) - E(AY)\} + (\pi - \pi^2) \delta^T
				E(X)}{\pi - \pi^2} \\
			&= - \frac{\pi\{\pi E(Y \mid A = 1) + (1 - \pi) E(Y \mid A =
				0)\} - \pi E(Y \mid A = 1)}{\pi - \pi^2}- \delta^T E(X)\\
			&= \beta_{\text{ATE}} - \delta^T E(X)
		\end{align*}
		as desired.
	\end{proof}
	
	\subsection{Proof of Lemma \ref{lemma:cons-an}}
\begin{proof}
To be able to use the results of \cite{white1980heteroskedasticity}, we need to check that our assumptions are sufficient to evoke regularity conditions cited by the author. Since $Y$ and $Z$ have bounded forth moments, there exist $\eta$ and $H$ such that 
$E(|\epsilon(\theta)^2|^{\eta+1})<H$ and
$E(|Z_{j}Z_{k}|^{\eta+1})<H$, $j,k=1,\dots,p$. By H\"{o}lder
inequality, this implies that
$E(\epsilon(\theta)^2|Z_{j}Z_{k}|^{\eta+1})<H$ is also
uniformly bounded. Furthermore, we assumed that $E(ZZ^t)$ is
positive definite, that is, $E(ZZ^t)$ is non-singular and $det
E(ZZ^t)>\eta>0$. The same is valid for
$E(\epsilon(\theta)^2ZZ^t)$. Thus, Assumptions 2 and 3 of
\cite{white1980heteroskedasticity} are satisfied and we can
use the same steps as the author to prove the consistency and
the asymptotic normality of $\hat{\theta}$.
\end{proof}
\subsection{Proof of Lemma \ref{lem:beta-var}}
\begin{proof}
First, consider the unrestricted case where $\Gamma = \Delta =
\mathbb{R}^p$. To use Lemma \ref{lemma:cons-an}, 
we simply need to
compute $E(ZZ^T)$ and $E(ZZ^T\epsilon^2)$ for
$Z=(1,A,X^T,AX^T)^T$. Employing $A\perp\!\!\!\perp X$, $E(Z)=0$,
$A^2=A$, and $E(XX^T)=\Sigma$, we have
\begin{equation*}
	E(ZZ^T)=
	\begin{pmatrix}
		1   & \pi & 0 & 0\\
		\pi & \pi & 0 & 0\\
		0   & 0   & \Sigma    & \pi\Sigma\\
		0   & 0   & \pi\Sigma & \pi\Sigma
	\end{pmatrix}.
\end{equation*}
Using properties of block diagonal matrices, it follows that
\begin{equation*}
	\begin{split}
		n\text{var}(\hat{\beta})\to
		\begin{bmatrix}
			\begin{pmatrix}
				1  & \pi \\
				\pi& \pi \\
			\end{pmatrix}^{-1}
			\begin{pmatrix}
				E(\epsilon^2)  & E(A\epsilon^2) \\
				E(A\epsilon^2)& E(A\epsilon^2)\\
			\end{pmatrix}
			\begin{pmatrix}
				1  & \pi \\
				\pi& \pi \\
			\end{pmatrix}^{-1}
		\end{bmatrix}_{(22)}
		=\frac{E\{(A-\pi)^2\epsilon^2 \}}{\pi^2(1-\pi)^2}.
	\end{split}
\end{equation*}
Here, $[\cdot]_{(22)}$ means the entry on the second row and second
column of the matrix.

If some dimensions of $\Gamma$ or $\Delta$ are singletons, we can
simply remove the corresponding entries in $Z$. By a similar
calculation, the same formula holds and the asymptotic variance of
$\hat{\beta}$ only differs in the regression error $\epsilon$, which
depends on $\theta = \theta(\Gamma,\Delta)$.
\end{proof}

\subsection{Proof of Lemma \ref{lemma:var-beta-jk-tilde}}\label{sec:proof_lemma4}
\begin{proof}
	We fix $\Gamma$ and $\Delta$ and suppress the dependence of
	$\hat{\beta}$ and $\tilde{\beta}$ on $(\Gamma,\Delta)$. We decompose
	$\tilde{\beta}$ as
	\[
	\tilde{\beta}=\hat{\beta}+\hat{\delta}^T\bar{X} =
	\hat{\beta}+\delta_s^T\bar{X}+(\hat{\delta}-\delta_s)^T\bar{X}.
	\]
	Due to the assumption $E(X) = 0$, we have $\hat{\delta}-\delta_s = O_p(n^{-1/2})$ and $\bar{X} = O_p(n^{-1/2})$ . Hence, the last term on the right hand side is negligible. Therefore,
	\[
	n\left\{\text{var}(\tilde{\beta})-\text{var}(\hat{\beta})\right\} \to
	n \text{var}(\delta_s^T \bar{X}) + 2 n \text{cov}(\hat{\beta}, \delta_s^T
	\bar{X}).
	\]
	
	Let $Z$ be the unrestricted variables in the linear regression. By
	applying the sandwich variance formula for the following set of
	estimating equations
	\[
	\label{eq:estimating-equation}
	\begin{split}
		E(X - \mu) = 0, \\
		E[Z \epsilon(\theta)] = 0,
	\end{split}
	\]
	we obtain
	\begin{align*}
		n\text{cov}\left(\hat{\theta},\bar{X} \right) &\to
		\left\{
		\begin{pmatrix}
			- I & 0 \\
			0 & - E(Z Z^T) \\
		\end{pmatrix}^{-1}
		\begin{pmatrix}
			\Sigma & E(X Z^T {\epsilon}) \\
			E(Z X^T {\epsilon}) & E(Z Z^T {\epsilon}^2) \\
		\end{pmatrix}
		\begin{pmatrix}
			- I & 0 \\
			0 & - E(Z Z^T) \\
		\end{pmatrix}^{-T}
		\right\}_{(2)} \\
		&= \left\{E(Z Z^T)  \right\}^{-1} E(Z X^T {\epsilon}),
	\end{align*}
	where $\epsilon$ is 
	${\epsilon}(\theta) = \epsilon(\theta_f)
	+ (\gamma_f - \gamma_s)^T X + A (\delta_f - \delta_s)^T X$. It follows that
	\begin{equation*}
		\begin{split}
			n\text{cov}\left(\hat{\beta},\bar{X}  \right) &\rightarrow \left[
			\left\{E(Z Z^T)  \right\}^{-1} E(Z X^T \epsilon)  \right]_{(22)} \\
			&= \left[
			\begin{pmatrix}
				1   & \pi & 0 \\
				\pi & \pi & 0 \\
				0 & 0 & * \\
			\end{pmatrix}^{-1}
			\begin{pmatrix}
				(\gamma_f - \gamma_s)^T \Sigma + \pi (\delta_f - \delta_s)^T
				\Sigma \\
				\pi (\gamma_f - \gamma_s)^T \Sigma + \pi (\delta_f -
				\delta_s)^T \Sigma\\
				*
			\end{pmatrix}
			\right]_{(22)} \\
			&= (\delta_f - \delta_s)^T \Sigma,
		\end{split}
	\end{equation*}
	where $*$ represent some unspecified matrices that are not important
	for deriving the quantities of interest. Therefore,
	\[
	\begin{split}
		n\left\{\text{var}(\tilde{\beta})-\text{var}(\hat{\beta})\right\} &\to
		n \text{var}(\delta_s^T \bar{X}) + 2 n \text{cov}(\hat{\beta}, \delta_s^T
		\bar{X}) \\
		&= \delta_s^T \Sigma \delta_s + 2 (\delta_f - \delta_s)^T \Sigma
		\delta_s \\
		&= \delta_s^T \Sigma (2 \delta_f - \delta_s).
	\end{split}
	\]
\end{proof}

\subsection{Variance orderings in Table \ref{tab:table-comparisons}} \label{sec:var-ord-table}
In this section, we provide an alternative, simpler proof to derive variance ordering in Table~\ref{tab:table-comparisons}. These conclusions can be derived from Theorems~\ref{theorem:var-order} and \ref{theorem:order-non-cent}.  Let  $\theta_f =
\theta_0(\mathbb{R}^p, \mathbb{R}^p)$ denote the full model
parameters. From the first order condition, we have
\begin{equation}\label{eq:gamma0_delta0_f}
\gamma_f= \Sigma^{-1} E(X Y\mid A=0), \quad \delta_f = \Sigma^{-1}
\{E(X Y\mid A=1)  - E(X Y\mid A=0) \}.
\end{equation}
Let $V_f = \lim_{n \to \infty} n \text{var}(\hat{\beta}_f)$ be the
asymptotic variance of $\hat{\beta}_f$.

(a) We compare variances of $\hat{\beta}$ in Model 1 with $Z_1 = (1, A, X^T)^T$ and in Model 2 with $Z_1 = (1, A)^T$. Consider $ \Gamma_1= \mathbb{R}^p  $, $ \Delta_1=0 $, $ \Gamma_2= 0$ and $ \Delta_2= 0$, that is, $ \hat\beta_1 $ is the ANCOVA and $\hat \beta_2 $ is the ANOVA estimator in Example 1 in the main document. In this case, only the third condition $\mathcal{U}(\Delta_1) \supseteq \mathcal{U}(\Gamma_1)$ in Theorem 1 is not satisfied. We show that $ V_1> V_2 $ when $\pi E(XY\mid A=0)  = (\pi - 1) E(X Y\mid A=1 ) $, $
\pi\neq \frac{1}{2} $,  and $ E(X Y\mid A=1 ) \neq  E(XY\mid A=0)
$.

By definition, $ \delta_1 =\delta_2= \gamma_2= 0 $, and  $ \gamma_1
= \Sigma^{-1} E(XY)= \gamma_f + \delta_f \pi$. Then, by applying derivations in Theorem 1, 
we have
\begin{align*}
\pi^2(1-\pi)^2( V_2 - V_f  )
&= E\big[ (A- \pi )^2 \{ \gamma_f^T X +A \delta_f ^T X\}^2  \big],\\
\pi^2(1-\pi)^2(         V_1 - V_f  )&= E\big[ (A- \pi )^2 \{
(\gamma_f- \gamma_1)^T X
+ A \delta_f ^T  X\}^2
\big],\\
&=  E\big[ (A- \pi )^2 (\gamma_1^T X )^2  \big] +  E\big[ (A- \pi )^2 (\gamma_f^T X+ A\delta_f^T X )^2  \big]\\
&\qquad  - 2E\big[ (A- \pi )^2 \gamma_1^TX (\gamma_f^T X+ A\delta_f^T X ) \big].
\end{align*}
Hence, using the fact that $ A\perp \!\!\! \perp X $,
\begin{align*}
&\pi^2(1-\pi)^2( V_1- V_2 ) \\
=& E\big[ (A- \pi )^2 (\gamma_1^T X
)^2  \big]  - 2E\big[ (A- \pi )^2
\gamma_1^TX (\gamma_f^T X+
A\delta_f^T X ) \big]\\
=& E \big[ (A - \pi)^2 (\gamma_1^T X) (\gamma_1 - 2 \gamma_f -
2 A \delta_f)^T X \big] \\
=& E\big[ (A- \pi )^2 (\gamma_f^T X + \pi \delta_f^TX ) ( \pi \delta_f^TX - \gamma_f^T X  - 2 A\delta_f^T X )  \big] \\
=& E\big[A (1- \pi )^2 (\gamma_f^T X + \pi \delta_f^TX ) (\pi \delta_f^TX - \gamma_f^T X  - 2 \delta_f^T X )  \big] \\
&\qquad +E\big[(1-A) \pi ^2 (\gamma_f^T X + \pi \delta_f^TX ) ( - \gamma_f^T X + \pi \delta_f^TX   )  \big] \\
=& \pi (1- \pi )^2 (\gamma_f + \pi \delta_f)^T \Sigma (\pi
\delta_f - \gamma_f - 2 \delta_f) +(1-\pi) \pi ^2 (\gamma_f + \pi \delta_f)^T \Sigma (\pi
\delta_f - \gamma_f) \\
=& \pi(1-\pi) (\gamma_f + \pi
\delta_f)^T \Sigma \{(3\pi - 2) \delta_f -
\gamma_f\}.
\end{align*}
When $ \gamma_f = (\pi - 1) \delta_f , \pi\neq \frac12$ and $
\delta_f\neq 0 $, we have
\begin{equation}\label{eq:anova-ancova}
\pi^2(1-\pi)^2( V_1- V_2)= \pi
(1-\pi) (2\pi-1)^2  E(\delta_f^2) >  0.
\end{equation}
Under this scenario, $(\tilde{V}_1, \tilde{V}_2)=(V_1, V_2)$. We can thus proceed in the same way to prove $\tilde{V}_1> \tilde{V}_2 $.

(b)  We compare variances of $\hat{\beta}$ in Model 1 with $Z_1 = (1, A,A X^T)^T$ and in Model 2 with $Z_1 = (1, A)^T$. First we shall prove $ V_2 \geqslant V_1$. Consider $ \Gamma_1= 0  $, $ \Delta_1=\mathbb{R}^p $, $ \Gamma_2= 0$ and $ \Delta_2= 0$.  In this case we have $\gamma_1 = \gamma_2 = \delta_2 = 0$, $\delta_1 = \Sigma^{-1}E(XY|A=1)$ and
\begin{equation}\label{eq:comparison_2}
\begin{split}
\pi^2(1-\pi)^2(V_2 - V_1)
=&E\left\{(A-\pi)^2
2\epsilon_1  A \delta_1^TX \right\} + E\left\{(A-\pi)^2
(A \delta_1^TX)^2\right\}\geqslant0.
\end{split}
\end{equation}
The first term on the right hand side in \eqref{eq:comparison_2} is $0$ by applying the sufficient condition in Theorem 1, 
$A \perp \!\!\! \perp X$ and $A^2 = A$.

Now we prove that $\tilde{V}_1 > \tilde{V}_2$ for some cases. Let $\gamma_f$ and $\delta_f$ be as defined in equation \eqref{eq:gamma0_delta0_f}, $\gamma_1 = \gamma_2 = \delta_2 = 0$ and $\delta_1 = \Sigma^{-1}E(XY|A=1) = \delta_f + \gamma_f\neq0$. In addition, let $\Omega_{l}=E(XY|A=l)$ where $l=0,1$. When $\Omega_0=-1/2\Omega_1$, we have
\begin{align}
\tilde{V}_1 - \tilde{V}_2
& \rightarrow V_1 + \delta_1^T\Sigma(2\delta_f-\delta_1) - V_2\nonumber\\
& =     \frac{E\{(A-\pi)^2 (A\delta_1^TX)^2\}}{\pi^2(1-\pi)^2}   + \delta_1^T\Sigma(2\delta_f-\delta_1)  \nonumber\\
& = \frac{E\{
(\gamma_f^TX+\delta_f^TX )^2
\}}{\pi}  + (\delta_f+\gamma_f)^T\Sigma(\delta_f-\gamma_f)  \nonumber\\
& =  \frac{(\delta_f+\gamma_f)^T\Sigma
(\delta_f+\gamma_f)+ \pi (\delta_f+\gamma_f)^T\Sigma(\delta_f-\gamma_f)}{\pi}  \nonumber\\
& = \frac{\Omega_1^T\Sigma^{-1}\Omega_1+ \pi \Omega_1^T\Sigma^{-1}(\Omega_1-2\Omega_0)}{\pi}  = \frac{\Omega_1^T\Sigma^{-1}\Omega_1+ 2\pi \Omega_1^T\Sigma^{-1}\Omega_1}{\pi} > 0 \label{eq:comparion3}.
\end{align}
\bibliography{literature-review}

\end{document}